\documentclass[journal]{IEEEtran}
\usepackage{amsthm}
\usepackage{amsmath}
\usepackage{amssymb}
\usepackage{floatfig,epsfig}
\usepackage{subfigure}
\usepackage{verbatim}
\usepackage{pifont}
\usepackage{algorithm}
\usepackage{algorithmic}
\usepackage{alltt}
\usepackage{bm}
\usepackage{isolatin1}
\usepackage{graphics}
\usepackage{epic}
\usepackage{eepic}
\usepackage{graphpap}
\usepackage{cite}
\usepackage{psfrag}

\newtheorem{col}{Corollary}
\newtheorem{pro}{Proposition}

\newtheorem{mydef}{Definition}

\newtheorem{lem}{Lemma}

\begin{document}

\title{Repair for Distributed Storage Systems \\in Packet Erasure Networks}

\author{Majid Gerami, Ming Xiao, Jun Li, Carlo Fischione, Zihuai Lin
\thanks{This work was supported partly by Wireless@KTH for the seed project,
entitled "Efficient content storage and dissemination in heterogeneous wireless
networks". Part of this paper was presented in IEEE International Conference on Communications, ICC, June 2013 \cite{Maj02}.

M. Gerami, M. Xiao, and C. Fischione are with the School of Electrical Engineering of KTH (The Royal Institute of Technology), Stockholm, Sweden (e-mail: \{gerami, mingx, carlofi\}@kth.se).

J. Li and Z. Lin are with the School of Electrical and Information Engineering, University of Sydney, Australia (e-mail: \{jun.li1, zihuai.lin\}@sydney.edu.au.).

}}


\maketitle
\begin{abstract}
\normalsize
Reliability   is essential for storing files in many applications of distributed storage systems. To maintain reliability, when a storage node fails, a new node should be regenerated by a repair process.  Most of the previous results on the repair problem assume perfect (error-free) links in the networks. However, in practice, especially in a wireless network, the transmitted packets (for repair) may be lost due to, e.g., link failure or buffer overflow.  We study the repair problem of distributed storage systems in packet erasure networks,  where  a packet loss is modeled as an erasure. The minimum  repair-bandwidth, namely the amount of information sent from the surviving nodes to the new node, is established under the ideal assumption of infinite number of  packet transmissions. We also study the bandwidth-storage tradeoffs in erasure networks. Then, the use of  repairing storage nodes (nodes with smaller storage space) is proposed to reduce the repair-bandwidth. We study the minimal storage of repairing storage nodes. For the case of a finite number of packet transmissions, the probability of successful repairing is investigated. We show that the repair with a finite number of  packet transmissions may use much larger bandwidth than the minimum  repair-bandwidth. Finally,  we propose a  combinatorial optimization problem, which results in the optimal repair-bandwidth  for the given packet erasure probability and finite packet transmissions.
\end{abstract}

%

\IEEEpeerreviewmaketitle

\section{Introduction}\label{sec:intro}
\IEEEPARstart{D}{istributed} storage systems have recently attracted substantial research interest for many applications such as file sharing, cloud storage and data centers. Although these research results are mostly for wired networks,  distributed storage systems can also be applied in wireless networks. For instance, consider a scenario of applying a distributed storage system in a delay tolerant network (DTN) with wireless channels, as shown in  Fig. \ref{Fig.1WirelessDTN}. A DTN is a  network in which there might not be direct links between a source and destinations, and  services can tolerate an acceptable level of incurred delay \cite{Jain01}. In the scenario  in Fig. \ref{Fig.1WirelessDTN},  a base station distributes a source file within the mobile nodes, which may move towards arbitrary directions.  Then a data collector (DC) can rebuild the source file by meeting  a certain number of these mobile storage nodes.

In distributed storage systems, it is important to keep the stored file reliable even if storage nodes are unreliable.  Unreliability in  storage nodes might stem from disk failure or power off. However, the failure might not be limited only to those cases. For instance, in the network in Fig.  \ref{Fig.1WirelessDTN},  a mobile node leaving the system can be considered as node failure. Since node failure may happen frequently, a distributed storage system needs a mechanism to maintain the reliability, namely, to regenerate a new node. Such a mechanism of a new node being regenerated is called a repair process. In this process, surviving nodes transmit sufficient packets to the new node for repair.

\begin{figure}
\centering
\psfrag{s}[][][3.5]{ $S$ }
\psfrag{A}[][][3.5]{DC }
\psfrag{node1}[][][2.5]{node 1 }
\psfrag{node2}[][][2.5]{node 2 }
\psfrag{node3}[][][2.5]{node 3 }
\psfrag{node4}[][][2.5]{node 4 }
\psfrag{mk}[][][3.0]{$M/k$ }
\resizebox{8cm}{!}{\epsfbox{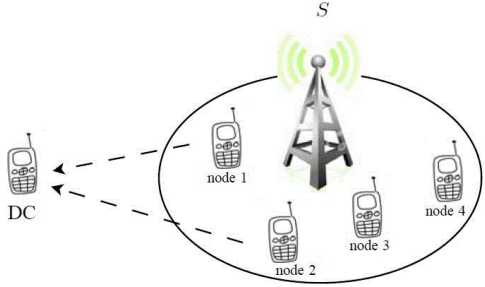}}
\caption{A distributed storage system in a wireless  network. The base station distributes a file among storage nodes on its coverage area. A data collector (DC) can recover the source file, even though it is out of base station coverage by meeting a sufficient number of storage nodes. In this example, the source file is coded by an ($n=4, k=2$)-MDS code. Thus, the DC by downloading from  any two  storage nodes can rebuild the source file.}
\label{Fig.1WirelessDTN}
\end{figure}

In wireline and wireless networks, there may also be unreliability in transmission links. In wireline networks, the transmitted packets may be lost due to  buffer overflow in an intermediate node \cite{Lun01}. In wireless networks, packets may be lost due to channel unreliability (caused by e.g., fading or interference). Thus,  not only storage nodes,  but also links between nodes can be unreliable. In such an environment, designing a robust distributed storage system is challenging.

We model a packet loss as an erasure, and then study the repair problem in packet erasure networks. We aim to minimize the amount of information that is sent from the surviving nodes to the new node, namely \emph{repair-bandwidth}, in packet erasure networks. In these networks, we derive a lower bound for the repair-bandwidth  and then  prove that the the lower bound is achievable by the use of random linear codes. We show that the minimum  repair-bandwidth can be achieved when the number of packet transmissions from each surviving node to the new node  tends to infinity (asymptotic analysis). By asymptotic analysis, we propose a method of exploiting limited-capacity storage nodes, termed as \emph{repairing storage nodes}, to reduce the repair-bandwidth. We study the minimal storage of repairing storage nodes.  We also show that if the number of packet transmissions is finite, the repair-bandwidth can be much larger than that of the asymptotic analysis. We term this repair-bandwidth as \emph{practical repair-bandwidth}.  Then we  propose a method to minimize the practical repair-bandwidth by studying the reliability of repair, i.e., the probability of successful repair. Our study shows that  a method with lower asymptotic repair-bandwidth is not necessarily  a better approach in reducing the practical repair-bandwidth and a combinatorial optimization problem can find the optimized approach based on the link packet erasure probabilities.

The  rest of the paper is structured as follows. In Section \ref{background}, the background  and the related works are discussed. Then in Section \ref{sec:Tradoff}, we study the fundamental optimal bandwidth-storage tradeoff in packet erasure networks.  In Section \ref{Sec:RepHelper}, we propose using repairing storage nodes  to reduce the repair-bandwidth. Next, in Section \ref{sec:ReduceFinite}, we study  methods to minimize the repair-bandwidth for the case of a finite number of packet transmissions. Finally we conclude the paper in Section \ref{sec:conclusion}.

\section{Background}\label{background}
Erasure codes are generally used for protecting files in  storage systems. For an $(n,k)$-erasure code, a source file is divided into $k$ parts (also called packets or fragments) and then encoded to $n$ parts. Suppose the source file, denoted by   vector $\mathbf{m}=[ m_1,...,m_k]^{T}$,\footnote{Superscript $T$ denotes transpose operation.} contains $k$ packets.  Each packet $m_i, (i \in [k]$\footnote{Notation $[k]$ denotes the set $\{1,\cdots,k\}$.}) consists of elements from the Galois field $\mathrm{GF(q)}$, where $q$ represents the  field size. Note that in this paper, to simplify illustration, we assume all packets have the same length and packets are the basic information unit.  If vector $\mathbf{x}=[X_1, X_2,\cdots, X_n]$ denotes the encoded vector then $X_i=\mathbf{g_i}^T\mathbf{m}, (i \in [n])$, where $\mathbf{g_i}$ is a $k$ dimensional column vector with elements from the Galois field $\mathrm{GF(q)}$. In the network coding literature, vector $\mathbf{g_i}$ is known as \textit{global encoding vector} \cite{Raymond}. Consequently, a data collector, which is interested to download the source file,  can reconstruct the file by receiving $k$ packets having independent global vectors, and then performing decoding e.g., Gaussian elimination to recover $\mathbf{m}$.

The maximum reliability for an $(n,k)$-erasure code is achieved when every $k$ out of $n$ packets can rebuild the source file. This property of reconstructing the file from any $k$ out of $n$ packets is termed as \emph{reconstruction property}. MDS (Maximum Distance Separable) codes are erasure codes satisfying the reconstruction property. By encoding a file of size $M$ by an $(n,k)$-MDS codes, each node stores $M/k$ packets in a distributed storage system, which is the minimum storage to reconstruct the file by $k$ nodes. As an example, consider the scenario in Fig. \ref{Fig.1WirelessDTN}, a file of size $M=4$ packets is encoded by a $(4, 2)$-MDS code. Thus each node stores $M/k=2$ packets and a sink (or a data collector) can rebuild the source file by meeting at least $k=2$ storage nodes.

When a node fails, to maintain the reliability of the system,  a new node is regenerated.  In the regenerating process, the surviving nodes transmit sufficient data to the new node such that the system with the new node still maintains the reconstruction property. Yet the new node may have different coded symbols comparing to the failed one.  This is called  \emph{functional repair}. In another kind of repair, which is known as  \emph{exact repair}, the new node is exactly the same as the failed node. In this paper, we only consider the functional repair. Yet, the results can be extended to the exact repair.

Although MDS erasure codes are efficient in the use of storage space for providing reliability, they are not efficient in the number of packets to download for repair, called \emph{repair-bandwidth}. A naive approach for repair  is to rebuild the source file (by downloading $M$ packets from the surviving nodes) and then regenerate the new node (e.g., exactly as the failed one). However, downloading $M$ packets for regenerating $M/k$ packets (which one node stores) is not efficient in term of the repair-bandwidth. Reference \cite{Dimk01} models the repair process by an information flow graph as a multicast problem in error-free networks (\textit{lossless networks}). Cut-set bound analysis on the information flow graph shows that  a sink (or data collector) can download the original file of size $M$ if
\begin{eqnarray}
\sum_{i=0}^{k-1} \min \{ \alpha, (d-i)\beta\} \geq M,
\label{Eq:bound}
\end{eqnarray}
where $\alpha$ denotes the individual node's storage capacity, and $d$ denotes the number of surviving nodes involving in the repair ($d \geq k$), and $\beta$ denotes the repair traffic sent by each of the nodes among $d$ surviving nodes.

\begin{mydef} Capacity of the distributed storage system:
 the term $\sum_{i=0}^{k-1} \min \{ \alpha, (d-i)\beta\}$ under the dynamic of node failure/repair. For the optimal codes $\sum_{i=0}^{k-1} \min \{ \alpha, (d-i)\beta\}=M$.
 \end{mydef}

Here the dynamic of node failure/repair means a process in which nodes continuously fail and repair.The authors in \cite{Dimk01} derived an explicit form of relation between $\alpha, \gamma=d\beta, d,$ and $k$ for the points on the fundamental optimal bandwidth-storage tradeoff (tradeoff between $\alpha$ and $\gamma$). The codes achieving the optimal tradeoff are called  \textit{regenerating codes}. The results in  \cite{Dimk01} show that by increasing the capacity of storage nodes  the repair-bandwidth can be reduced.

The codes achieving two extreme points on the fundamental bandwidth-storage tradeoff are termed as \textit{minimum storage regenerating (MSR)} and \textit{minimum bandwidth regenerating (MBR)} codes. These two points can also be derived by two sequential optimization processes under constraint (\ref{Eq:bound}).  MSR codes are achieved by first minimizing the storage and then minimizing the repair-bandwidth. The minimum storage capacity required for the reconstruction property is $M/k$. Thus, storage nodes by MSR codes store the same amount of data as the MDS codes. However, MSR codes have the minimum bandwidth in regenerating a new node. From the bound in (\ref{Eq:bound}), we can derive the minimum  repair-bandwidth for an MSR code as
\begin{eqnarray}
\alpha_{\mathrm{MSR}}=&\frac{M}{k}, \nonumber  \\
\gamma_{\mathrm{MSR}}=&\frac{Md}{k(d-k+1)}.
\label{Eq-MSR1}
\end{eqnarray}

In the optimization process, if we first minimize the repair-bandwidth and then storage per each node, another extreme point, the MBR point, is achieved. It can easily be verified that in general $\gamma \geq \alpha$ \cite{ShahJ12}. For  MBR codes $\gamma=\alpha$. Therefore, setting  $\gamma=d\beta=\alpha$ on the optimum bound $\sum_{i=0}^{k-1} \min \{ \alpha, (d-i)\beta\} = M$  yields
\begin{eqnarray}
\alpha_{\mathrm{MBR}}=&\frac{2Md}{k(2d-k+1)}, \nonumber  \\
\gamma_{\mathrm{MBR}}=&\frac{2Md}{k(2d-k+1)}.
\label{Eq-MBR1}
\end{eqnarray}

References \cite{Wu01}, \cite{Wu02}, \cite{Rashmi01} studied the code construction and achievablity of the functional and exact repair. In \cite{Anne01}, cooperative regenerating codes are proposed to reduce the bandwidth in the scenario of multiple-node failure. References \cite{Maj01, Maj03} suggest surviving node cooperation in order to minimize the cost of repair in  multi-hop networks. Yet, in most of the previous work of regenerating codes, it is assumed that links between storage nodes are perfect, without any error or erasure. In distributed storage systems, especially for those in wireless networks (as the example in Fig. 1), packets on the channels may be lost due to physical layer errors (e.g., channel fading and interference) or network layer errors (e.g., buffer overflow). Then, the redundant data needs to be transmitted for repair. Recently reference \cite{Rashmi02} has suggested a regenerating code which is resistant to a specific number of path failures by requesting more nodes to join the repair process. Particularly,  for the code resistant to $d_2$ number of path failures, it is required to transmit from $d_{\mathrm{tot}}=d_1+d_2$ surviving nodes instead of $d_1$ nodes ($d_1$ nodes are assumed to be sufficient for repair with the perfect channels). However, reference \cite{Rashmi02} has not considered transmission errors of individual packet and how to achieve optimal repair-bandwidth under such scenarios. We call a regenerating process  successful when the new node together with the surviving nodes has the reconstruction property. Also, reference \cite{Rashmi02} has not studied the probability of successful repair and how to construct the optimal codes (to reduce the repair bandwidth). We will study the optimal codes to achieve the minimal repair bandwidth. We shall consider the probability of successful repair and show that  the optimal $d_1 \text{ and } d_2$ depend on the erasure probability of the links. Thus, we can find the code maximizing the probability of successful repair, given the constraints of repair-bandwidth.

Other related results on network coding for erasure networks are as follows.  The capacity of wireless erasure  networks has been studied in \cite{Dana01}.  It is shown in \cite{Lun01} that the capacity of packet erasure networks can be achieved by random linear codes.  In references \cite{Jain02}, \cite{Leong01}, the probability of successful reconstruction of a source file is studied in erasure networks. However, the papers did not study the regenerating problem. References \cite{Jain01}, \cite{Jain02} study delay tolerant networks. In \cite{Jain02}, the authors show that there is no unique answer to the question whether coding  for distributing a file  maximizes the probability of successful reconstruction or not.


For the notation, we follow  the same notations as  \cite{Dimk01}. That is, in a distributed storage system characterized by  parameters $(n, k, d, \alpha, \gamma, M)$, $n$ denotes  the number of nodes; $k$ denotes the reconstruction parameter (every $k$ nodes can rebuild the original file); $\alpha$ denotes the storage capacity of each node; $d$  denotes  the number of surviving nodes helping the repair (here considering only $k\leq d \leq n-1$);  $\gamma$ denotes the repair-bandwidth and $M$ denotes the size of the source file. The notations are also listed in Table~\ref{table:notation}. In this paper, we assume $\alpha$, $\beta$ and $\alpha/\beta$ are positive integers.

\begin{table}[t]
\caption{Notations} 
\centering 
\begin{tabular}{c l } 
\hline\hline 
Symbol & Definition \\[0.5ex] 
\hline 
$n$ & the number of storage nodes in a distributed storage\\ & system\\ 
$k$ & the minimum number of nodes needed to  reconstruct \\ &the original file \\
$d$ & the number of surviving nodes for a repair process ($k\leq d \leq n-1$) \\
$\alpha$ &  individual node storage capacity \\
$\beta$ & number of packets from each  surviving node \\ & in the repair process \\
$\gamma$ & repair-bandwidth ($\gamma=d\beta$)\\
$M$ & size of the source  file \\
[1ex]
\hline 
\end{tabular} \label{table:notation} 
\end{table}

\section{Repair in erasure networks}\label{sec:Tradoff}

In this section, we study the repair process in a packet erasure network. The system model is as follows. When a node fails, a new node is regenerated by the help of $d$ surviving nodes. There are direct but lossy links from the surviving nodes to the new node. The transmitted packets on the links might be erased with a packet erasure probability $p$.  To simplify analysis, we assume all the links have equal packet erasure probabilities. In addition, we assume that the erasure events for different packets are independent identically distributed (i.i.d). When a packet is erased all the content of the packet will be lost. Otherwise  all the content will be received correctly. We shall analyze the required repair-bandwidth in packet erasure networks  in what follows.

For a point-to-point channel with i.i.d. packet erasure and probability $p$, the capacity of channel is $1-p$, which is derived from the capacity of binary erasure channels as described in \cite{Cover01}. The capacity is achieved when the number of packet transmissions tends to infinity. That is, to correctly receive $N$  packets in a packet erasure channel with an erasure probability $p$, a transmitter must send $N/(1-p)$ packets, if $N$ tends to infinity.
\begin{figure*}
 \centering
 \psfrag{S}[][][2.5]{ $S$ }
  \psfrag{newnode}[][][2.5]{ new node}
  \psfrag{DC}[][][2.5]{ $DC$ }
  \psfrag{a}[][][2.5]{ $\alpha$ }
   \psfrag{b}[][][2.5]{ $\beta^{'}, p$ }
  \psfrag{cut}[][][2.5]{ $Cut$ }
   \psfrag{infty}[][][2.5]{ $\infty$ }
   \psfrag{vdot}[][][2.5]{ $\vdots$ }
   \psfrag{incdot}[][][2.5]{ $\ddots$ }
    \psfrag{in}[][][2.5]{$in$ }
    \psfrag{out}[][][2.5]{$out$ }
         \psfrag{stage0}[][][1.5]{stage 0 }
        \psfrag{stage1}[][][1.5]{stage 1 }
         \psfrag{hdots}[][][1.5]{$\ldots $}
 \resizebox{10cm}{!}{\epsfbox{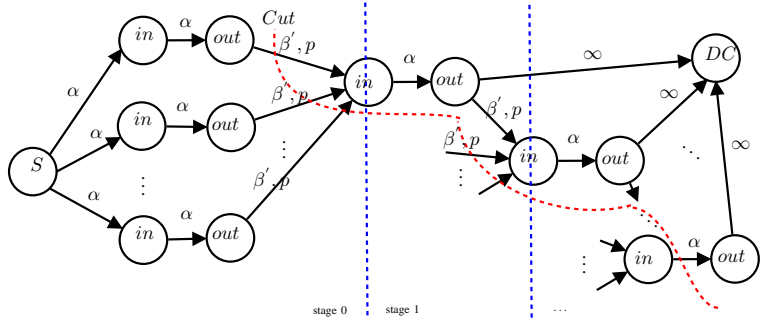}}
\caption{Information flow graph for distributed storage systems in packet erasure networks. Each of $d$ surviving nodes transmits $\beta^{'}$ packets to the new node. The links from the surviving nodes to the new node have equal packet erasure probability, $p$. Due to the erasures on the links,  the new node receives  $(1-p)\beta^{'}$ in average, from each link.}
 \label{Erasureflowgraph}
\end{figure*}

For  multicast in a packet erasure network, an upper bound of information rate from a source to all the destinations (or data collectors)  can be derived by  cut-set bound analysis \cite{Cover01, Raymond}. A cut refers to dividing the set of nodes in the network into two complement sets of $U$ and $\overline{U}$ such that one set contains the source and another contains sinks. The cut capacity is the sum of capacities of the edges from the set containing the source to the complement set. From the cuts, a cut with the minimum capacity, termed as \emph{min-cut}, determines the maximum rate of information from a source to destinations.  The maximum rate can be achieved asymptotically, for instance, by the use of random linear network codes \cite{Lun02}. Although the rate can also be achieved by retransmission when perfect feedback channels are available, it may be difficult due to the incurred delay or complexity of feedback in multicast networks \cite{Lun02}.

The maximum information rate between a source and data collectors in a distributed storage system under the dynamic of node failure/repair can be derived by the cut-set analysis on the information flow graph. The information flow graph is a directed acyclic graph on which a source file is transmitted to the data collectors in the presence of (potentially) infinite number of node failure/repair. Each storage node is modeled by two nodes, $in$ and $out$ nodes, which are connected by a link of capacity $\alpha$. A source node, denoted as $S$, contains a file of size $M$ packets. The source is connected to $in$ nodes  by links of capacity $\alpha$. When a node fails, $d$ surviving nodes send packets to the new node.  The erasure channel from each surviving node to the new node is represented by a link with capacity $(1-p)\beta^{'}$, where $\beta^{'}$ denotes  the number of packets transmitted from each node among $d$ surviving nodes, and $p$ denotes the packet erasure probability. A data collector (DC) by connecting at-least $k$ nodes (through infinite-capacity links) can reconstruct the original file. An information flow graph is shown in Fig. \ref{Erasureflowgraph}. The following proposition characterizes the capacity of a distributed storage system in packet erasure networks.

\begin{pro} The capacity of a distributed storage system with  parameters $(n, k, d, \alpha, \beta^{'}, M)$ having (potentially) infinite number of node failure/repair, and links with packet erasure probability $p$ is  given by
\begin{eqnarray}
\mathcal{C}=\sum_{i=0}^{k-1} \min\{\alpha, (d-i)(1-p)\beta^{'} \}.
\end{eqnarray}
\end{pro}
\begin{proof}
We show that the upper bound of the information rate in such a multicast network equals $\sum_{i=0}^{k-1} \min\{\alpha, (d-i)(1-p)\beta^{'} \}$. For that, we follow the approach adopted  in \cite{Dimk01}. We thus show that the min-cut in the information flow graph equals $\mathcal{C}$. Since the information flow graph is direct acyclic graph, then there is a topological sorting \cite{Raymond}. Hence, the information flow graph after $k$  failure/repair processes can be shown in $k$ subsequent stages, as depicted in Fig. \ref{Erasureflowgraph}. We first prove there is a cut in the network with capacity $\mathcal{C}$. For the purpose, consider a cut that passes  a route with a minimum capacity at any stage of repair. For example at stage $1$, the cut selects a route between $\alpha$ and $d(1-p)\beta^{'}$, where $(1-p)\beta^{'}$ is the capacity of an erasure link between one of the surviving nodes and the new node. At stage $2$, since the new node can get $(1-p)\beta^{'}$ packets from previously regenerated node, then the cut selects between $\alpha$ and $(d-1)(1-p)\beta^{'}$, and so on. Finally, there will be a graph with a cut capacity equivalent to $\sum_{i=0}^{k-1} \min\{\alpha, (d-i)(1-p)\beta^{'} \}$. Any other cut has capacity greater than, or equal to, $\mathcal{C}$. This states $\mathcal{C}$ as the min-cut of the network.

For the achievability,  random network codes can be used for encoding the data on the storage and also for transmitting the repair traffic from the surviving nodes to the new node. Since random network coding is capacity approaching for the multicast in packet erasure networks \cite{Lun02}, we can use the  codes in our proof. Then the capacity is achievable, conditioning that the field size of  codes is large enough, and the number of packet transmissions tends to infinity.
\end{proof}

Alternatively, for the achievability, we can construct the regenerating codes in packet erasure networks by first constructing a regenerating code $(n, k, d, \alpha, \beta=\beta^{'}(1-p), M)$ in a lossless network. For  given $\alpha$ and $\beta^{'}$ on the optimal bound, the constraint $\sum_{i=0}^{k-1} \min \{ \alpha, (d-i)(1-p)\beta^{'} \} =M$ yields $\sum_{i=0}^{k-1} \min \{ \alpha, (d-i)\beta \} = M$. Hence, we can construct an $(n,k,d,\alpha,\beta,M)$ regenerating codes, e.g. as in \cite{Dimk01},\cite{Wu01}. Next,  we use random linear network coding for transmitting packets from $d$ surviving nodes to the new node. That is, each surviving node linearly combines its $\beta$ repair packets by  random coding coefficients, which are uniformly selected from $\mathrm{GF(q)}$.  The new node recovers $\beta$ packets  from each link after receiving $\beta$ packets having independent global vectors. For that, each surviving node transmits $\beta/(1-p)$ packets.

Consequently, it turns out that the asymptotic repair-bandwidth (when the number of repair packets goes to infinity), which is denoted as $\gamma^{'}$, for a given $\alpha$  equals $\gamma^{'}=d\beta^{'}=d\beta/(1-p)$. The optimal repair-bandwidth for varying individual node storage capacities can thus be derived by multiplying the  factor $1/(1-p)$ to the fundamental repair-bandwidth in lossless network (derived in \cite{Dimk01}). This is expressed in the following corollary.

\begin{col} \label{Proposition:Regbound}
Let $\gamma^{'}$ denote the repair-bandwidth in a packet erasure network. Then, the storage capacity of each node, $\alpha$ is derived  as  a function of parameters $n,k,d,\gamma^{'},p$ as
\begin{eqnarray}
\alpha=
\begin{cases} \frac{M}{k} & \text{  if    } \gamma^{'} \in \left[\frac{f(0)}{1-p},+ \infty \right), \\
\frac{M-g(i)(1-p) \gamma^{'}}{k-i} & \text{  if    } \gamma^{'} \in \left[\frac{f(i)}{1-p},\frac{f(i-1)}{1-p}\right),\nonumber \\& \text{  } i=1,...,k-1,  \end{cases}
\label{Eq-Regbound}
\end{eqnarray}
where
\begin{eqnarray}
f(i)=\frac{2Md}{(2k-i-1)i+2k(d-k+1)},
\end{eqnarray}
and
\begin{eqnarray}
g(i)=\frac{(2d+2k+i+1)i}{2d}.
\end{eqnarray}
\end{col}
Hence, two extreme points on the bandwidth-storage tradeoff are as follows. For the MSR codes,
 \begin{eqnarray}
\alpha_{\mathrm{MSR}}=&\frac{M}{k}, \nonumber\\
\gamma_{\mathrm{MSR}}=&\frac{Md}{k(d-k+1)(1-p)}.
\label{Eq-MSR-p}
\end{eqnarray}
For the MBR codes,
 \begin{eqnarray}
\alpha_{\mathrm{MBR}}=&\frac{2Md}{k(2d-k+1)}, \nonumber\\
\gamma_{\mathrm{MBR}}=&\frac{2Md}{k(2d-k+1)(1-p)}.
\label{Eq-MBR-p}
 \end{eqnarray}

Fig. \ref{tardeoffs} shows the fundamental bandwidth-storage tradeoffs in packet erasure networks with $p=0.1,0.2,0.3$. Expectedly, the larger erasure probability leads to the higher repair traffic. These fundamental bandwidth-storage tradeoffs are derived with the assumption that the number of packet transmission in repair tends to infinity.
\begin{figure}
 \centering
 \psfrag{data1data1data1}[][][2.5]{ $p=0.1$ }
 \psfrag{data2data2data2}[][][2.5]{ $p=0.2$ }
 \psfrag{data3data3data3}[][][2.5]{ $p=0.3$ }
 \psfrag{ylabel}[][][2.5]{ $\alpha$ }
 \psfrag{xlabel}[][][2.5]{ $\gamma^{'}$ }
 \psfrag{title}[][][2.5]{ }
 \resizebox{8cm}{!}{\epsfbox{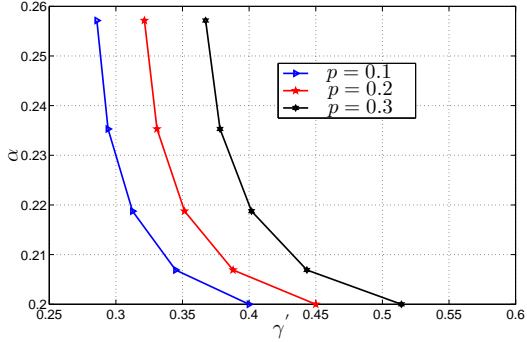}}
\caption{Fundamental bandwidth-storage tradeoffs in erasure networks, for $n=10, k=5, d=9$, and for different packet erasure probabilities. Higher packet erasure probability ($p$), higher repair traffic on average. The storage and repair-bandwidth in this figure are normalized for M=1.}
 \label{tardeoffs}
\end{figure}

\section{Reducing asymptotic repair-bandwidth}\label{Sec:RepHelper}
\subsection{Problem Formulation}
Following the results in the previous section,  we shall propose the approach of using capacity-limited storage nodes in  networks to reduce the repair-bandwidth. We show that this method can reduce the repair bandwidth in lossless networks and also in packet erasure networks. As shown before, the asymptotic repair-bandwidth for a distributed storage system with parameters $(n,k,d,\alpha,\gamma^{'}=d\beta^{'},M)$ in a packet erasure network with link erasure probability $p$ is computed by multiplying the factor $1/(1-p)$ to the repair-bandwidth  for a system with parameters $(n,k,d,\alpha,\gamma=d\beta,M)$ in a lossless network. Namely, $\gamma^{'}=\gamma/(1-p)$.

In what follows, we will first show that the repair-bandwidth is a decreasing function of $d$, the number of surviving nodes in the repair. The following two lemmas show this for two extreme points, namely MSR and MBR codes.

\begin{lem} Consider  a distributed storage system  using an MSR code with parameters $(n,k,d,\alpha=M/k,\gamma(d)=d\beta,M)$ in a lossless network. If the number of surviving nodes for repair increases from $d$ to $d+h$, for a positive integer $h$, then the  repair bandwidth decreases from $\gamma_{\mathrm{MSR}}(d)=Md/(k(d-k+1)$ to $\gamma_{\mathrm{MSR}}(d+h)=M(d+h)/(k(d+h-k+1)$.
\end{lem}
\begin{proof} From (\ref{Eq-MSR1}), we have $\gamma_{\mathrm{MSR}}(d)=Md/(k(d-k+1)$, and $\gamma_{\mathrm{MSR}}(d+h)=M(d+h)/(k(d+h-k+1)$. We shall show that $\gamma_{\mathrm{MSR}}(d) \geq \gamma_{\mathrm{MSR}}(d+h)$.
Since $k \geq 1$, we have
\begin{eqnarray}
0  \geq & h(-k+1) \\
\Rightarrow d^2+hd-kd+d   \geq & d^2+hd-kd+d-hk+h \\
\Rightarrow d(d+h-k+1)  \geq & (d+h)(d-k+1)\\
\Rightarrow \frac{Md}{k(d-k+1)} \geq & \frac{M(d+h)}{k(d+h-k+1)}\\
\Rightarrow \gamma_{\mathrm{MSR}}(d) \geq & \gamma_{\mathrm{MSR}}(d+h)
\end{eqnarray}
\end{proof}

Similarly, we can prove for the MBR codes.

\begin{lem} Consider  a distributed storage system  using an MBR code with parameters $(n,k,d,\alpha_1=\gamma_1=2Md/k(2d-k+1),M)$ in a lossless network. If the number of surviving nodes for repair increases from $d$ to $d+h$, for  a positive integer $h$, then the  repair bandwidth decreases from $\gamma_{\mathrm{MBR}}(d)=2Md/k(2d-k+1)$ to $\gamma_{\mathrm{MBR}}(d+h)=2M(d+h)/k(2(d+h)-k+1)$.
\end{lem}
\begin{proof} From (\ref{Eq-MBR1}), we have $\gamma_{\mathrm{MBR}}(d)=2Md/k(2d-k+1)$, and $\gamma_{\mathrm{MBR}}(d+h)=2M(d+h)/k(2(d+h)-k+1)$. We shall show that $\gamma_{\mathrm{MBR}}(d) \geq \gamma_{\mathrm{MBR}}(d+h)$.
Since $k \geq 1$, we have
\begin{eqnarray}
0  \geq & h(-k+1) \\
\Rightarrow 2d^2+2hd-kd+d   \geq & 2d^2+2hd-kd+d-hk+h \\
\Rightarrow d(2(d+h)-k+1)  \geq & (d+h)(2d-k+1)\\
\Rightarrow \frac{2Md}{k(2d-k+1)} \geq & \frac{2M(d+h)}{k(2(d+h)-k+1)}\\
\Rightarrow \gamma_{\mathrm{MBR}}(d) \geq & \gamma_{\mathrm{MBR}}(d+h)
\end{eqnarray}
\end{proof}

For other codes, a similar result can also be shown.
To reduce the repair-bandwidth by increasing $d$, we can use more storage nodes with capacity $\alpha$ for repair in a  distributed storage system with parameters $(n, k, d, \alpha, \gamma, M)$. For illustration, we call the storage with capacity $\alpha$ as a \emph{complete storage node}.  Using an additional complete storage node for repair requires a storage node with capacity $\alpha$. It can be quite expensive if $\alpha$ is large. If we could use additional storage nodes (for repair) with considerably smaller capacity, but using the same repair-bandwidth as using additional complete storage nodes, then the system could save considerably in storage costs.  We call these storage nodes, each of which  with capacity $\alpha^{'} < \alpha$ as \emph{repairing storage nodes}. Note that the repairing storage nodes are only used for repair. In other words, data collectors do not contact them for the original file. A natural question is that what is the minimum $\alpha^{'}$  such that  a repairing storage node is replaced by a complete storage node.

To address the question, we formulate the problem as follows. Let us consider a  distributed storage system (DSS) containing $n$ complete storage nodes, each of which  with capacity $\alpha$, and $h$ repairing storage nodes, each of which  with capacity $\alpha^{'} < \alpha$. A file of size $M$ is divided into $k$ parts and coded by an $(n,k)$-erasure code, and distributed among $n$ complete storage nodes such that any $k$ of $n$ complete storage nodes can rebuild the original file (reconstruction property). We assume repairing storage nodes are also directly connected to the source node and are only used for repair. Thus, we assume that repairing storage nodes always exist (They seldom fail and   a new repairing storage node is regenerated by directly downloading data from the source node if it fails).  When a complete storage node fails,  $d$ out of $n$ complete storage nodes help to regenerate the new node. In addition,   $h$ repairing storage nodes also help the repair.  In the repair each of these $d+h$ helper nodes sends $\beta$ packets to the new node. Thus, the repair-bandwidth, denoted as $\gamma$, is computed by $\gamma=(d+h)\beta$. The DSS is represented by DSS$(n,k,d,h,\alpha,\alpha^{'},\beta,M)$.

The optimization problem seeks to find the minimum amount for $\alpha^{'}$, for $\alpha^{'} < \alpha$,  such that the repairing storage nodes can still be useful  in the repair of the complete storage nodes. Formally, for given $d, h, $ and $\alpha$, we first find a lower bound for the repair-bandwidth ($\gamma$) by the following optimization problem

\begin{eqnarray}
\min_{\beta}  & \gamma=(d+h)\beta \nonumber \\
\text{subject to:}  & \sum_{i=0}^{k-1} \min \{ \alpha, (d+h-i)\beta\} \geq M.
\label{Formula}
\end{eqnarray}
The constraint in this optimization problem is for the reconstruction property. Next, we find the minimum amount of $\alpha^{'}$ for achieving the lower bound of the repair-bandwidth found in  (\ref{Formula}).

\begin{figure*}
 \centering
  \psfrag{newnode}[][][2.5]{ new node}
  \psfrag{DC}[][][2.5]{ $DC$ }
  \psfrag{s}[][][2.5]{ $S$ }
  \psfrag{a}[][][2.5]{ $\alpha$ }
  \psfrag{b}[][][2.5]{ $\beta, \varepsilon$ }
  \psfrag{cut}[][][2.5]{ $Cut$ }
  \psfrag{infty}[][][2.5]{ $\infty$ }
  \psfrag{vdots}[][][2.5]{ $\vdots$ }
  \psfrag{incdot}[][][2.5]{ $\ddots$ }
  \psfrag{in}[][][2.5]{$in$ }
  \psfrag{out}[][][2.5]{$out$ }
  \psfrag{stage0}[][][1.5]{stage 0 }
  \psfrag{stage1}[][][1.5]{stage 1 }
  \psfrag{hdots}[][][1.5]{$\ldots $}
  \psfrag{d1}[][][1.5]{$d\beta$}
  \psfrag{d2}[][][1.5]{$(d-1)\beta$}
  \psfrag{dk}[][][1.5]{$(d-k+1)\beta$}
  \psfrag{ap}[][][1.5]{$\alpha^{'}$}
  \psfrag{x}[][][2.5]{A repairing storage node}
  \psfrag{Cut}[][][1.5]{Cut}
 \resizebox{10cm}{!}{\epsfbox{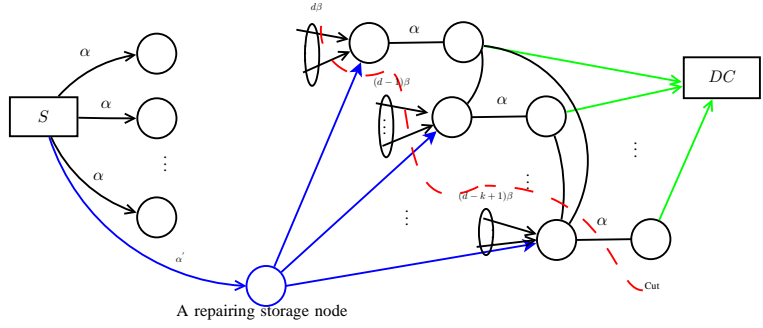}}
\caption{Information flow graph for distributed storage systems with a repairing storage node.}
 \label{Fig:HelperNode}
\end{figure*}

\subsection{The bandwidth gain of repairing storage nodes}\label{Sec:lowerBound}
Consider a distributed storage system with parameters $(n,k,d,h,\alpha,\alpha^{'},\beta,M)$.  Each helper node transmits $\beta$ packets for repairing of each complete storage node.  The lower bound of the repair-bandwidth is found by the cut-set bound analysis in the information flow graph.

 The information flow graph consists of four types of nodes: a source node, complete storage nodes with capacity $\alpha$, repairing storage nodes with capacity $\alpha^{'}$, and data collectors (DCs). Complete storage nodes are generally modeled by two nodes as $in$ and $out$ nodes, which are connected by a link of capacity $\alpha$. A source node, denoted as $S$, initially transmits information to the storage nodes through the links of capacity $\alpha$. The source node also transmits $\alpha^{'}$ packets of information to each repairing storage node. A DC recovers the original file by connecting to at least $k$ complete storage nodes. If a complete storage node fails, $d$ surviving nodes plus $h$ repairing storage nodes transmit in total $\gamma=(d+h)\beta$ to the new node. The information flow graph  of repair process is illustrated in Fig. \ref{Fig:HelperNode}.

The following proposition evaluates  an upper bound in the information rate in the presence of $h$ repairing storage nodes.

\begin{pro}   An upper bound of information rate in distributed storage system with parameters $(n,k,d,h,\alpha,\alpha^{'},\beta,M)$ is

\begin{equation} \sum_{i=0}^{k-1} \min \{ \alpha, (d-i+h)\beta \},\label{Formula:min-cut+h} \end{equation}
and thus DC can recover the original file if \begin{equation} \sum_{i=0}^{k-1} \min \{ \alpha, (d-i+h)\beta \} \geq M \label{Formula:min-cut+h2}\end{equation}
 \label{Pro:Lowerbound+h} \end{pro}
\begin{proof}
By cut-set  analysis, we show that $\sum_{i=0}^{k-1} \min \{ \alpha, (d-i+h)\beta \}$ is the min-cut in the information flow graph. Consider a DC that is connected to $k$ new nodes.   Since the information flow graph is a directed acyclic graph, there is a topological sorting \cite{Raymond}. Thus the information flow graph after $k$  failure/repair processes can be shown in $k$ subsequent stages, as depicted in Fig. \ref{Fig:HelperNode}. The min-cut at any stage passes a route with a minimum capacity. For example at stage $1$ the min-cut selects between $\alpha$ and $(d+h)\beta$. At stage $2$, since the new node can get $\beta$ amount of information from regenerated node in previous stage, then the min-cut selects between $\alpha$ and $[(d+h)-1]\beta$, and so on. Finally a graph with a cut capacity equals to $\sum_{i=0}^{k-1} \min \{ \alpha, (d-i+h)\beta \}$ is yielded. Any other cut in the network has capacity greater than, or equal to the calculated min-cut. To transmit a file of size $M$, from the source to a DC, there must be $\sum_{i=0}^{k-1} \min \{ \alpha, (d-i+h)\beta \}\geq M$.
\end{proof}

 Solving  inequality (\ref{Formula:min-cut+h2}) for $M$ yields a lower bound of $\beta$ in term of $M$. We note that the lower bound has been derived with the assumption that the repairing storage nodes store a proper coded data such that can be used the same as complete storage nodes in repair. We shall discuss the achievable codes in next subsection.  The potential gain of using the repairing storage nodes can be derived at this stage.  We can find the minimum storage regenerating codes and minimum bandwidth regenerating codes when there exist $h$ repairing storage nodes. To find the minimum repair-bandwidth in the MSR point, we first minimize $\alpha$. Since $\alpha \geq \frac{M}{k}$, then $\alpha_{\mathrm{MSR}}=M/k$. For the minimum storage to satisfy $\sum_{i=0}^{k-1} \min \{ \alpha, (d-i+h)\beta \}=M$, we must have $(d+h-k+1)\beta= \alpha$, thus the minimum repair-bandwidth is,
  \begin{eqnarray}
\gamma_{\mathrm{MSR}}=\frac{M(d+h)}{k(d-k+h+1)}.
\label{Eq-MSR+h}
\end{eqnarray}
Another extreme point, the MBR point, is given for $\gamma=\alpha$ in inequality (\ref{Formula:min-cut+h2}).  Hence,
 \begin{eqnarray}
\alpha_{\mathrm{MBR}}=&\frac{2M(d+h)}{k(2(d+h)-k+1)},\\
\gamma_{\mathrm{MBR}}=&\frac{2M(d+h)}{k(2(d+h)-k+1)}.
\label{Eq-MBR+h}
 \end{eqnarray}
This results show a reduction in repair-bandwidth by a factor of $d(d+h-k+1)/((d+h)(d-k+1))$ for MSR codes, and by a factor of $d(2(d+h)-k+1)/((d+h)(2d-k+1))$ for MBR codes.  We note that we could achieve these gains by adding $h$ complete storage nodes. However repairing storage nodes have smaller storage space ($\alpha^{'}<\alpha$). The minimal $\alpha^{'}$ for $h=1$ will be discussed as follows. For $h>1$, we can follow a similar approach.

\subsection{Minimum repairing storage capacity for the MSR codes}\label{Sec:MSR}
We show that there exist linear codes for repair in a DSS satisfying bound in Proposition \ref{Pro:Lowerbound+h} when each complete storage node stores ($\alpha=M/k$), and  a repairing storage node stores $\alpha^{'}=\beta$. We note that $\alpha^{'}=\beta$ is the minimum storage capacity  for repairing storage node, since each repairing storage node transmits $\beta$  packets of information.

\begin{pro} [\emph{Achievability for MSR codes with $\alpha^{'} = \beta$}] For a repair process in a DSS with parameters $(n,k, d=n-1, h=1, \alpha=M/k, \alpha^{'}= \beta, \beta, M)$ of MSR codes,  there exist linear codes if each of $d+1$ helper nodes, including the repairing storage node, transmits  $\beta=M/(k(d-k+2))$ packets to the new node.
\label{pro:AchievabilityMSR}\end{pro}

\begin{proof} See Appendix A.
\end{proof}
Proposition \ref{pro:AchievabilityMSR} shows that  reducing the repair-bandwidth for the MSR codes can be attained by adding a repairing storage node with storage capacity of $\beta$ instead of using a complete storage node with capacity $\alpha=M/k=\beta(d-k+2)$. Hence, it requires less storage by the ratio of $(d-k+2)$.

In the next subsection, we investigate the minimum storage required for a repairing storage node at the MBR point.

\subsection{Minimum repairing storage capacity for the MBR codes}\label{Sec:MBR}
We show that another extreme point on the bandwidth-storage tradeoff can be achieved in the presence of one repairing storage node. As we shall show, for the achievability of the MBR point, the repairing storage node requires more storage capacity than in the case of MSR codes (which was $\alpha^{'} =\beta$). Next proposition proves non-achievability of the MBR point for $\alpha^{'}<k\beta$.

\begin{lem} [\emph{Non-achievability for MBR codes with $\alpha^{'} < k\beta$}] For the MBR codes  in a DSS with parameters $(n,k,d,h=1,\alpha=(d+1)\beta,\alpha^{'},\beta,M)$, we have $\alpha^{'} \geq k\beta$.
\end{lem}
\begin{proof} We prove by contradiction. Assume $\alpha^{'} =(k-1)\beta < k\beta$ (without essential loss of generality we assume $\alpha^{'}$ and $\beta$ are non-negative integers). Let $S$ denote the random variable representing the source file. For the source file of size $M$ bits, we have $H(S)=M$, where $H(X)$ refers to the entropy of random variable $X$. Next, let $W_{l}$ denote the random variable representing the content of node $l$ for $l \in [n]$. Assume $W_{n+i}$ denotes the corresponding  random variable for the content of the new node after $i$ stages of repair. Since every $k$ nodes have to reconstruct the original file, we have,
\begin{eqnarray}
M=&H(W_{l+1},W_{l+2},\cdots,W_{l+k}) \label{Formul:r1}\\
=&H(W_{l+1})+H(W_{l+2}\mid W_{l+1})+\cdots\nonumber \\&+H(W_{l+k}\mid W_{l+1},\cdots,W_{l+k-1} ) \label{Formul:r2}\\
=&(d+1)\beta+(d)\beta+\cdots\nonumber \\&+(d-k+1)\beta+(d-k+1)\beta \label{Formul:r3}\\
=&M-\beta,\label{Formul:r4}
\end{eqnarray}
which is a contradiction.

In the proof, (\ref{Formul:r1}) follows by the chain rule of entropy, (\ref{Formul:r2}) follows from the fact that each new node in stage $i$, conditioning on knowing information in previous stages, receives $(d+1-i)\beta$ new information. Also, (\ref{Formul:r3}) follows from the fact that $\alpha^{'} =(k-1)\beta$, and (\ref{Formul:r4}) follows the fact that the point is located on the bandwidth-storage tradeoff, thus $\sum_{i=0}^{k-1} (d-i+1)\beta \} = M$.\end{proof}

 The next proposition shows that  MBR codes can be achieved if the repairing storage node stores $\alpha^{'}= k\beta$.
 \begin{pro} [\emph{Achievability for MBR codes with $\alpha^{'} = k\beta$}] For the repair process of MBR codes with a repairing storage node with storage capacity  $\alpha^{'}= k\beta$, there exist linear codes if each of $d+1$ nodes, including the repairing storage node, transmits  $\beta=2M/(k(2d-k+3))$ packets to the new node.
\label{pro:AchievabilityMBR}\end{pro}
\begin{proof} See Appendix B.\end{proof}
Proposition \ref{pro:AchievabilityMBR} shows that reducing the repair-bandwidth for the MBR codes can be attained by adding a repairing storage node with storage capacity of $k\beta$ instead of $(d+1)\beta$. Note that in some scenarios (if not most), $d = n -1$. Thus, using repairing storage nodes always reduces storage space substantially.

According to results in Section \ref{sec:Tradoff}, we can derive the asymptotic optimal repair-bandwidth  in packet erasure networks as follows.

\subsection{Repair with repairing storage nodes in packet erasure networks}
Above results show that repair bandwidth can be reduced if a repairing storage node is used for a distributed storage system with error-free channels. In what follows, we shall show that a repairing storage node can also reduce the repair-bandwidth in packet erasure networks. The analysis and code construction are similar with the analysis in Section \ref{sec:Tradoff}. Again, in a packet erasure network with links having equivalent erasure probabilities, $p$, the asymptotic optimal repair-bandwidth equals to that of the  repair-bandwidth in lossless network multiplying $1/(1-p)$.  The following two corollaries illustrate more formally the impact of a repairing storage node in packet erasure networks. We note that the results are for $h=1$, and for $h>1$, we can follow the similar approach.
\begin{col} For a DSS with  parameters $(n,k,d,h=1,\alpha=M/k,\alpha^{'} =\beta,\beta^{'}=\beta/(1-p),M)$ in a packet erasure network with channels having packet erasure probability $p$, the asymptotic optimal repair-bandwidth repair is  $\gamma_{\mathrm{MSR}}^{'}=M(d+1)/(k(d-k+2)(1-p))$.
\end{col}
Similarly, for MBR codes there is a corollary as follows.
\begin{col} For a DSS with  parameters $(n,k,d,h=1,\alpha,\alpha^{'} =k\beta,\beta^{'}=\beta/(1-p),M)$ in an erasure network with channels having packet erasure probability $p$, the asymptotic optimal repair-bandwidth  is  $\gamma_{\mathrm{MBR}}^{'}=2M(d+1)/(k(2d-k+3)(1-p))$.
\end{col}

\section{Reducing the repair-bandwidth for the case of finite packet transmissions} \label{sec:ReduceFinite}

Above we have investigated the optimal repair-bandwidth in packet erasure networks under the ideal  assumption of infinite number of repair packet transmissions.  However, due to the constraints of e.g., delay or complexity \cite{Ming11}, \cite{Koller11}  the number of packet transmissions may be finite. It is thus interesting to study the repair-bandwidth in the case of a finite number of packet transmissions. Since the repair may fail (the node does not receive sufficient repair packets), we first analyze the probability of successful repair (PSR). Then, we will propose a method to reduce the repair-bandwidth under the constraint of the PSR. The asymptotic optimal repair-bandwidth derived in previous sections will be used as a lower-bound of the repair-bandwidth of a finite number of packet transmissions.

We assume that the packet erasures are i.i.d. Bernoulli random process with probability $p$ and a regenerating code designated by  parameters $(n, k, d, \alpha, \beta, M)$.  A surviving node then transmits $t$ packets formed by linear combination (in  $\mathrm{GF(q)}$) of $\beta$ repair packets.  Previously we showed that $t=\beta(1-p)$ for the infinite number of packet transmissions. For a finite number of transmissions, let $P_{\beta}$ denote the probability of successfully receiving $\beta$ repairing packets from each surviving node. If $t$ packets are transmitted, the probability of successful recovery of $\beta$ packets (using the result in \cite{Acedanski01},\cite{Martal01}) equals
\begin{eqnarray}
P_{\beta}=\begin{cases} \sum_{i=\beta}^{t} \binom{t}{i} (1-p)^{i}p^{(t-i)}\frac{\prod_{l=0}^{\beta-1}(q^{i}-q^l)}{q^{\beta i}} & \text{ if } t\geq \beta,\\ 0 & \text{otherwise}. \end{cases}
\label{EQ:P_beta}
\end{eqnarray}
Hence, the PSR is equivalent to the probability of successfully receiving repair traffic from $d$ links. That is
\begin{eqnarray}
P_{s}=P_{\beta}^d.
\label{EQ:P_beta}
\end{eqnarray}

Fig. \ref{Fig:3D-Ps} shows the PSR for a distributed storage system using MBR codes with parameters $(n=10,k=5,d=9,\alpha=18, \beta=2, M=70)$ and $p=0.3$.  We observe that the required repair-bandwidth for $P_s$ approaching 1  is larger than the optimum repair-bandwidth. That is, the bandwidth overhead ratio $t/\beta$ might be several times greater than the optimal value $1/(1-p)$ due to finite number of packet transmissions. In what follows, we shall find the minimum bandwidth under the constraint of the PSR. Let us define $\delta$, for $0 \leq \delta \leq 1$, as a parameter indicating how close  the PSR is to $1$. In addition, let us define the minimum required bandwidth to achieve the PSR greater than $1-\delta$ as \emph{practical repair-bandwidth} for  given $\delta$ and $d$ helper nodes, which is denoted as $\widehat{\gamma}(\delta,d)$. Then
\begin{eqnarray}
\widehat{\gamma}(\delta,d)= \min_t & dt \nonumber\\
\text{subject to:}   & P_s \geq 1-\delta.
\label{EQ:beta}
\end{eqnarray}

\begin{figure}
 \centering
 \psfrag{zlabel}[][][2.5]{ $P_s$ }
 \psfrag{xlabel}[][][2.5]{ $\frac{t}{\beta}$ }
 \psfrag{ylabel}[][][2.5]{ $\beta$ }
 \psfrag{o}[][][2.5]{ oooooooppppp }
 \resizebox{8cm}{!}{\epsfbox{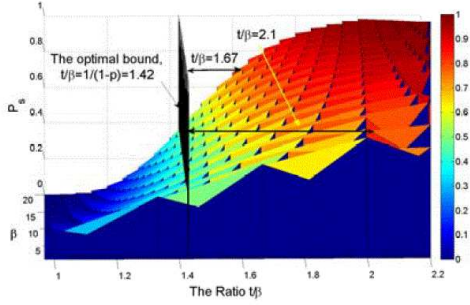}}
\caption{Probability of successful repair in erasure networks over different values of $\beta$ and bandwidth overhead ratio $t/\beta$. When the number of packet transmissions tends to infinity, $t/\beta \rightarrow 1/(1-p)$ for successful repair. However for a finite number of packet transmissions, larger number of packets ($t$) have to be transmitted for successful repair. The figure also compares the practical repair bandwidth for $P_s=0.9$ (regarding $\delta=0.1$). The larger $\beta$, the lower the bandwidth overhead ratio, $t/\beta$.}
 \label{Fig:3D-Ps}
\end{figure}

Then, we shall propose a method to reduce the practical repair-bandwidth. We will first show that the problem has a combinatorial optimization structure.  For illustration, we give an example, where two repair schemes are compared. Consider a distributed storage system with parameters $(n=10,k=5, d_{\mathrm{tot}}=9, M=70)$, where $d_{\mathrm{tot}}$ denotes the total number of surviving nodes in the repair. Suppose there are two schemes of exploiting these $9$ surviving nodes in the repair. For the first scheme, a regenerating code  $(n,k,d_{\mathrm{tot}},\alpha,\gamma(d_{\mathrm{tot}}),M)$ is used and then each of $d_{\mathrm{tot}}$ surviving nodes sends $t$ packets formed by linear combinations of $\beta(d_{\mathrm{tot}})$ repair packets. In this scheme, failure in receiving $\beta(d_{\mathrm{tot}})$ packets from only one link causes failure in the repair. In the second scheme, a regenerating code $(n,k,d_1,\alpha,\gamma(d_{1}),M)$ for $d_1< d_{\mathrm{tot}}$ is used, and each of $d_{\mathrm{tot}}$ surviving nodes send $t$ packets formed by linear combinations of $\beta(d_1)$ repair packets. Successfully receiving $\beta(d_{1})$ packets from $d_1$ out of $d_{\mathrm{tot}}$ links is sufficient for the repair.  The second scheme requires larger asymptotic repair bandwidth ($d_1< d_t$ thus $\gamma^{'}(d_1) \geq \gamma^{'}(d_{\mathrm{tot}}$)). Let $d_1$ denote the number of surviving nodes to the aim of reducing asymptotic repair-bandwidth, and $d_2$ denote the number of surviving nodes sending redundant data.   For the first scheme, we use an MBR code for $d_1=9,d_2=0$. From (\ref{Eq-MBR-p}), we can calculate the asymptotic repair-bandwidth $\gamma^{'}_1=25.65$ (corresponding to $\beta^{'}_1=2.85$). In the second scheme, we use an MBR code for parameter $d_1=7$ and then $d_2=2$. From (\ref{Eq-MBR-p}), we have $\gamma^{'}_2=36$ (corresponding to $\beta^{'}_2=4$). In the second scheme, two surviving nodes transmit redundant data such that successful receiving packets from at least $7$ out of these $9$ links yield successful repair. In this case, the PSR can be evaluated as
\begin{eqnarray}
 P_s=\sum_{i=7}^{9} \binom{9}{i} (P_{\beta})^i (1-P_{\beta})^{9-i}.
 \end{eqnarray}
 The PSR for these two schemes have been compared in Fig. \ref{Fig:comparePracticalBandwidth}. Moreover, the practical repair-bandwidth for $\delta=0.01$ has been compared between these two methods. We see the scheme with smaller asymptotic optimal repair-bandwidth has almost two times larger practical repair-bandwidth than the other scheme. 

In a general case, consider a repair process in a distributed storage system. Packets on the links are erased i.i.d. with a  probability $p$. Successful receiving $\beta$  packets from each of $d_1$ links guarantees successful repair. Note that if the new node receives fewer than $\beta$  packets for a helper node, the received packets from this nodes cannot be used. To increase the PSR, $d_2$ number of surviving nodes transmit redundant data for the repair. In total there are $d_{\mathrm{tot}}$ surviving nodes. Each surviving node still transmits  $t$ packets, each of which is formed by a linear combination of $\beta(d_1)$ repair packets. Hence, the PSR is that the new node receives  from at least $d_1$  out of $d_1+d_2$ surviving nodes, which is calculated by
 \begin{eqnarray}
 P_s=\sum_{i=d_1}^{d_1+d_2} \binom{d_1+d_2}{i} (P_{\beta})^i (1-P_{\beta})^{d_1+d_2-i}.
 \end{eqnarray}
 Given the constraint that the PSR is greater than $1-\delta$, we minimize the practical repair-bandwidth $\widehat{\gamma}(\delta,d_1+d_2)$ by changing the value of $d_1$ and $d_2$. The optimization problem can be formulated as follows,
\begin{eqnarray}
\min \limits_{d_1,d_2} &  \widehat{\gamma}(\delta,d_1+d_2) \\
\text{subject to:} & P_s \geq 1-\delta,\\
& d_1+d_2 \leq d_{\mathrm{tot}}.
\label{max-ps}
\end{eqnarray}

\begin{figure}
 \centering
 \psfrag{o1}[][][1.5]{asymptotic optimal for $d_1=9$ }
  \psfrag{o2}[][][1.5]{asymptotic optimal for $d_1=7$}
   \psfrag{p1}[][][1.5]{practical for $d_1=7$,$d_2=2$ }
   \psfrag{p2}[][][1.5]{practical for $d_1=9$,$d_2=0$ }
   \psfrag{xlabel}[][][2.0]{ $t$ }
    \psfrag{ylabel}[][][2.0]{ $P_s$ }
    \psfrag{data1datadatadatadata}[][][1.5]{ $P_s \text{ for } d_1=9$ }
    \psfrag{data2datadatadatadata}[][][1.5]{ $P_s \text{ for } d_1=7, d_2=2$ }
 \resizebox{8cm}{!}{\epsfbox{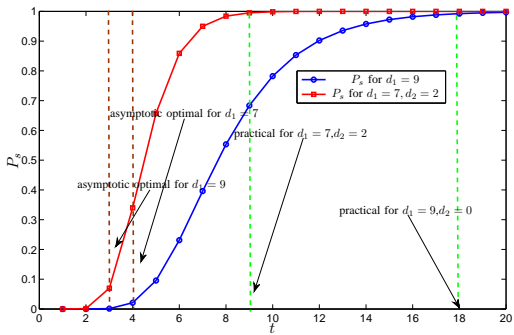}}
\caption{Probability of successful repair in erasure networks for $n=10, k=5$ over the repair-traffic from each node. The figure compares practical repair-bandwidth regarding $\delta=0.01$ versus optimal repair-bandwidth. }
 \label{Fig:comparePracticalBandwidth}
\end{figure}

Solving the optimization problem on the previous example shows that there is no  unique solution for all cases. That is,  the optimal repair-bandwidth approach depends on  the probability of packet erasure on the links. For illustration, we use the optimization problem in the previous example and find the corresponding $d_1$ and $d_2$  for the different values of packet erasure probabilities. The values of $d_1$ and $d_2$ that minimize the repair-bandwidth are shown in Fig. \ref{Fig.optimum_d_and_d2}. We can conclude that for the network with higher erasure probabilities, more helper nodes (larger $d_2$) should be used to increase the probability of successful repair.  Conversely, for the network with lower erasure probabilities,  less redundant data is needed, and the optimal practical repair-bandwidth is closer to the optimal asymptotic repair-bandwidth.

\begin{figure}
 \centering
 \psfrag{datadatadatadata1}[][][2.0]{ $d_1$ }
 \psfrag{datadatadatadata2}[][][2.0]{ $d_2$ }
 \psfrag{title}[][][1.5]{}
  \psfrag{ylabel}[][][2.5]{ number of storage nodes  }
 \psfrag{xlabel}[][][3]{ $p$  }
\resizebox{8cm}{!}{\epsfbox{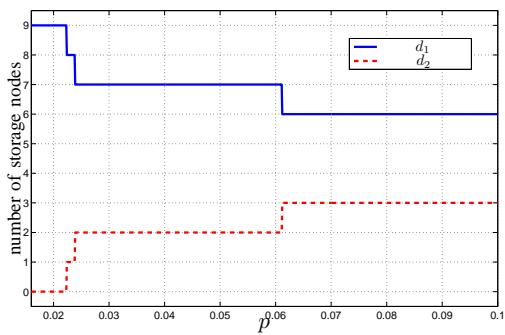}}
\caption{The value of $d_1$ and $d_2$ that minimize the practical repair-bandwidth over different values of links packet erasure probability, $p$.  Here, $\delta=0.0001$, and $p$ changes from $0.01$ to $0.1$. For a network with higher $p$ the reliability of repair becomes more critical and then  repair-bandwidth is minimized when the  redundant information is increased ($d_2$ increases).}
 \label{Fig.optimum_d_and_d2}
\end{figure}
\section{Conclusions} \label{sec:conclusion}

We studied the regeneration problem for distributed storage systems where channels are unreliable.
We investigated the storage-bandwidth tradeoff and the optimal  repair-bandwidth for packet erasure networks. We showed that repairing storage nodes can reduce the repair-bandwidth in packet erasure networks. We also studied the minimal storage for the repairing storage nodes.  We investigated the probability of successful repair and the approach to reduce the repair-bandwidth when number of transmitted packets is finite. We showed that the optimal repair-bandwidth depends on the channel-erasure probabilities. In this paper, we assume that channels have non-bursty packet losses. Studying  the repair problem in a network with bursty packet losses can be an interesting topic for future work.


\appendix
\subsection{Proof of Proposition \ref{pro:AchievabilityMSR}}
The  proof is based on random linear coding. By exploiting  sparse-zero lemma \cite{Raymond}, we show for large finite field size there exist linear codes for the repair problem.
\begin{lem}[sparse-zero lemma]\label{SPZ}  Consider a
multi-variable polynomial $g(\alpha_1,\alpha_2,...,\alpha_n)$ which is not identically  zero, and has the maximum degree
  in each variable at most $d_0$. Then, there exist variables  $\gamma_1,\gamma_2,...,\gamma_n$ in the finite field $\mathrm{GF(q)}$, for $q \geq d_0$, such that $g(\gamma_1,\gamma_2,...,\gamma_n)\neq 0$.
 \end{lem}
\begin{proof} See proof of Lemma 19.17 in \cite{Raymond}. \end{proof}
 We shall give the code construction with the minimum repair-bandwidth. We split the source file of a size $M$ into $k$ fragments. We denote the source file by vector $\mathbf{x}=[x_1, x_2, \cdots, x_{k(n-k+1)}]^T$. Substituting these set of parameters $(d=n-1,h=1,M=k(n-k+1))$ in Eq. (\ref{Eq-MSR+h}), we have $\beta=1$ . We construct an $(n,k)\texttt{-MDS}$ code using the following Vandermonde matrix, as a generator matrix, $G$,
\begin{eqnarray}
\mathbf{G}=  \left(\begin{array}{ccccc}
1 & \alpha_1 & \alpha_1^2 & \cdots & \alpha_1^{c-1}\\
1 & \alpha_2 &  \alpha_2^2 & \cdots & \alpha_2^{c-1}\\
\vdots & \vdots &  \vdots & \ddots & \vdots\\
1 & \alpha_r &  \alpha_r^2 & \cdots & \alpha_r^{c-1}\\
\end{array}\right),\end{eqnarray}
where  $r=n(n-k+1)+1$, $c=k(n-k+1)$, and  $\alpha_i$s for $i \in \{ 1,\cdots, r\}$ are distinct elements from the finite field $\mathrm{GF(q)}$. $q$ is the design parameter and should be select properly (for the Vandermonde matrix $q\geq r$). Now, without loss of generality, assume rows $i$ till $i+n-k$ in matrix $\mathbf{G}$ represent the code on node $i$, and denoted as $\mathbf{Q}_i$. Also the last row in matrix $\mathbf{G}$ represent the code for the repairing storage node, and is denoted as $\mathbf{q}_1$.

Suppose node $s_1$ fails. We show that the  code on new node, denoted as $\mathbf{Q}_{s_1}^{'}$, can be constructed such that,
\begin{eqnarray}
\det([\mathbf{Q}^{'}_{s_1}, \mathbf{Q}_{s_2},\cdots,\mathbf{Q}_{s_k}]\neq 0.
\label{Formula:nonZeroPoly}
\end{eqnarray}
By the construction of matrix $\mathbf{G}$, the code vectors of selecting $k-1$ nodes (let say nodes $s_2,\cdots, s_{k-1}$) among $n$ storage nodes (that results $(n-k+1)(k-1)$ vectors), and one vector from $d=n-k+1$ surviving nodes (that results $d-k+1=n-k$ vectors) plus one vector from the repairing storage node, which results in total $(n-k+1)(k-1)+(n-k)+1=k(n-k+1)$ vectors, are full rank. This means that, the left-hand side in equation (\ref{Formula:nonZeroPoly}) is not identically zero. Consequently,  we can use Lemma \ref{SPZ} to deduce that there exist linear codes for large enough finite field size. This finalizes the proof.

\subsection{Proof of Proposition \ref{pro:AchievabilityMBR}}
Let the $M$ dimensional vector $\mathbf{s}$ denote the source file. Also, let $\mathbf{x}_i$ and $\mathbf{Q}_i$  denote the vector containing data on node $i$ and the code on node $i$, respectively. Then $\mathbf{x}_i=\mathbf{Q}_i \mathbf{s}$. Matrix $\mathbf{Q}_i$ has $\alpha \times M$ size and $\mathbf{x}_i$ is a vector of size $\alpha$. We first construct an MBR code in a DSS with parameters $(n+1,k,d+1,\alpha=(d+1)\beta,\beta=2M/(k(2d-k+3)),M)$. The first part is without any repairing storage node, and the code can be constructed based on the path weaving approach in \cite{Wu01}. Then, we store in $n$ storage nodes the designed code. In the second step, we store proper data on the repairing storage node. We store $k\beta$ coded packets  on the memory of repairing storage node in this way: assuming we are in repair stage $i$, for $i \in [k]$, we store on the repairing storage node $\beta$ repair packets corresponding to the failure in stages $i=1,\cdots,k$. That means at repair stage $i$ ($i \in [k]$), the repairing storage node only transfers $\beta$ packets related to the corresponding failed node. For stages $i > k$, the repairing storage node sends a linear combination of its $k \times \beta$ stored data vectors.

We shall show that any repairing traffic that one node sent at stage $i> k$ can be stated as a linear combination of  repairing traffic that the same node have already sent in stages $i=1,\cdots,k$. In other words, the repair traffic in any stage is a linear combination of  $k\beta$ stored data. To prove that, consider that every $k$ nodes can recover the source file. Hence, we can state the code on any node by a linear combination of codes on nodes $i=1,\cdots,k$. Since repairing traffic is also a linear combination of codes on storage nodes, then the repair traffic one node sent in stage $i>k$ can be stated as a linear combination of data vectors in stages $i=1,\cdots,k$. This finalizes the proof.

\end{document}